\documentclass[runningheads]{llncs}
\pdfoutput=1 
\usepackage{cite,mathptmx,graphicx,hyperref,amstext,color}

\newif\ifSubmission
\Submissionfalse

\newif\ifProceedings
\Proceedingsfalse

\newif\ifArxiv
\Arxivtrue

\let\oldendproof\endproof
\def\endproof{\qed\oldendproof}

\def\comment#1{}%
\def\withcomments{%
  \usepackage{color}
  \newcounter{mycommentcounter}%
   \def\comment##1{\refstepcounter{mycommentcounter}%
    \ifhmode%
     \unskip%
     {\dimen1=\baselineskip \divide\dimen1 by 2 %
       \raise\dimen1\llap{\tiny
  {-\themycommentcounter-}}}\fi%
     \marginpar[{\renewcommand{\baselinestretch}{0.8}%
       \hspace*{-2em}\begin{minipage}{12em}\footnotesize%
[\themycommentcounter]:%
\raggedright ##1\end{minipage}}]{\renewcommand{\baselinestretch}{0.8}%
       \begin{minipage}{12em}\footnotesize%
[\themycommentcounter]: \raggedright%
##1\end{minipage}}}%
  }
\newcommand{\david}[1]{\comment{\textcolor{green}{\textbf{DE:}} #1}}
\definecolor{orange}{rgb}{1,0.5,0}

\newcommand{\martin}[1]{\comment{\textcolor{red}{\textbf{MN:}} #1}}

\newcommand{\red}{\ensuremath{\text{red}}}
\newcommand{\green}{\ensuremath{\text{green}}}

\title{Optimal 3D Angular Resolution for Low-Degree Graphs}

 \author{David Eppstein\inst{1} \and Maarten L\"offler\inst{1} \and Elena Mumford\inst{2} \and Martin N\"ollenburg\inst{1}}

 \institute{\noindent
 \inst{1}Department of Computer Science, University of California, Irvine, USA\\
 \inst{2}Department of Mathematics and Computer Science, TU Eindhoven, The Netherlands}

\begin{document}
\maketitle

\begin{abstract}
  We show that every graph of maximum degree three can be drawn in
  three dimensions with at most two bends per edge, and with
  $120^\circ$ angles between any two edge segments meeting at a vertex
  or a bend. We show that every graph of maximum degree four can be
  drawn in three dimensions with at most three bends per edge, and
  with $109.5^\circ$ angles, i.\,e., the angular resolution of the
  diamond lattice, between any two edge segments meeting at a vertex
  or bend.
\end{abstract}

\section{Introduction}

Much past research in graph drawing has shown the importance of avoiding sharp angles at vertices, bends, and crossings of a drawing, as they make the edges difficult to follow~\cite{HuaHonEad-PacVis-08}. There has been much interest in finding drawings where the angles at these features are restricted, either by requiring all angles to be at most $90^\circ$ (as in orthogonal drawings~\cite{EigFekKla-DG-01} and RAC drawings~\cite{AngCitDiB-GD-09,DidEadLio-WADS-09,DujGudMor-09}) or more generally by attempting to optimize the \emph{angular resolution} of a drawing, the minimum angle that can be found within the drawing~\cite{CarEpp-GD-06,GarTam-ESA-94,GutMut-GD-97,Mal-STOC-92}.

Three-dimensional graph drawing opens new frontiers for angular resolution in two ways. First, in three-dimensional graph drawing, there is no need for crossings, as any graph can be drawn without crossings; however, finding a compact layout that uses few bends and avoids crossings can sometimes be challenging. Second, and more importantly, in 3d there is a much greater variety in the set of ways that a collection of edges can meet at a vertex to achieve good angular resolution, and the angular resolution that may be obtained in 3d is often better than that for a two-dimensional drawing. For instance, in 3d, six edges may meet at a vertex forming angles of at most $90^\circ$, whereas in 2d the same six edges would have an angular resolution of $60^\circ$ at best.

The problem of optimizing the angular resolution of a collection of edges incident to a single vertex in 3d is equivalent to the well-known \emph{Tammes' problem} of placing points on a sphere to maximize their minimum separation; this problem is named after botanist P.~M.~L.~Tammes who studied it in the context of pores on grains of pollen~\cite{Tam-RTBN-30}, and much is known about it~\cite{ClaKep-PRSL-86}.
For graphs of degree five or six, the optimal angular resolution of a three-dimensional drawing is $90^\circ$, as above, achieved by placing vertices on a grid and drawing all edges as grid-aligned polylines.  The simplicity of this case has freed researchers to look for three-dimensional orthogonal drawings that, as well as optimizing the angular resolution, also optimize secondary criteria such as the number of bends per edge, the volume of the drawing, or combinations of both~\cite{BieThiWoo-Algo-06,EadSymWhi-DAM-00,Woo-TCS-03}. Thus, in this case, it is known that the graph may be drawn with at most 3 bends per edge in an $O(n)\times O(n)\times O(n)$ grid and with $O(1)$ bends per edge in an $O(\sqrt n)\times O(\sqrt n)\times O(\sqrt n)$ grid~\cite{EadSymWhi-DAM-00}. For graphs of maximum degree five a tighter bound of two bends per edge is also known~\cite{Woo-TCS-03}; a well known open problem asks whether the same two-bend-per-edge bound may be achieved for degree six graphs~\cite{Dem-TOPP-02}.

In two dimensions, $90^\circ$ angular resolution is optimal for graphs
that may be nonplanar, because every crossing has an angle at least
this sharp. However, in 3d, crossings are no longer a concern and
graphs of degree three and four may have angular resolution even
better than $90^\circ$. In particular, in the \emph{diamond lattice},
a subset of the integer grid, the edges are parallel to the long
diagonals of the grid cubes and meet at angles of
$\arccos(-1/3)\approx 109.5^\circ$, the optimal angular resolution for
degree-four graphs (Figure~\ref{fig:spacefillers}, left). For graphs
with maximum degree three, the best possible angular resolution at any
vertex is clearly $120^\circ$; three edges with these angles are
coplanar, but the planes of the edges at adjacent vertices may differ:
for instance, Figure~\ref{fig:spacefillers}(right) shows an infinite space-filling graph in which all vertices are on integer grid points, all edges form face diagonals of the integer grid, and all vertices have $120^\circ$ angular resolution.

\begin{figure}[t]
\centering
\includegraphics[width=2.25in]{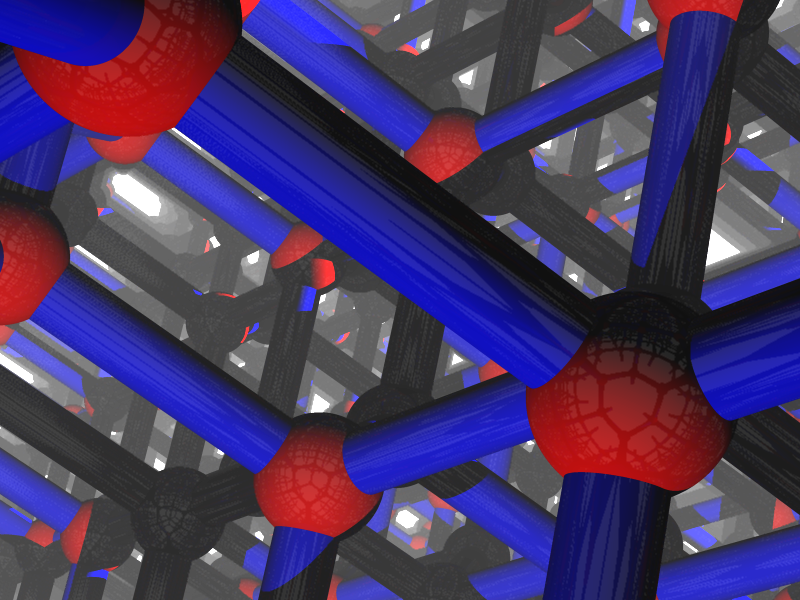}\qquad
\includegraphics[width=2.25in]{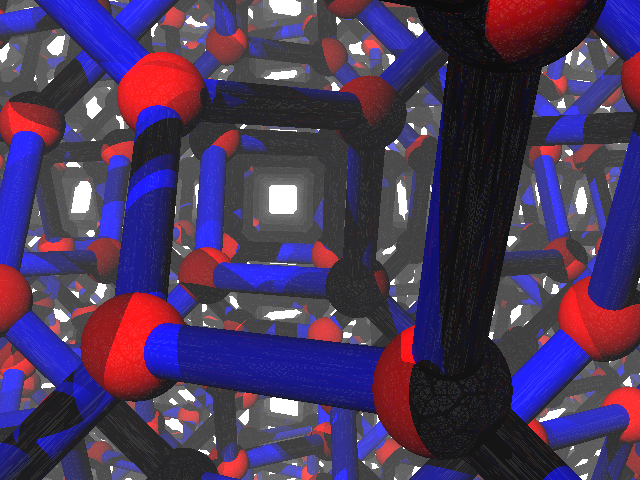}
\caption{Left: The three-dimensional diamond lattice, from~\cite{Epp-GD-08}. Right: A space-filling 3-regular graph with $120^\circ$ angular resolution.}
\label{fig:spacefillers}
\end{figure}

The primary questions we study in this paper are how to achieve optimal $120^\circ$ angular resolution for 3d drawings of arbitrary graphs with maximum degree three, and optimal $109.5^\circ$ angular resolution for 3d drawings of arbitrary graphs with maximum degree four.  We define angular resolution to be the minimum angle at any bend or vertex, matching the orthogonal drawing case, and we do not allow edges to cross. These questions are not difficult to solve without further restrictions (just place the vertices arbitrarily and use polylines with many bends to connect the endpoints of each edge) so we further investigate drawings that minimize the number of bends, align the vertices and edges of the drawing with the integer grid similarly to the alignment of the spacefilling patterns in Figure~\ref{fig:spacefillers}, and use a small total volume. We show:
\begin{itemize}
\item Any graph of maximum degree four can be drawn in 3d with optimal $109.5^\circ$ angular resolution with at most three bends per edge, with all vertices placed on 
  an $O(n)\times O(n)\times O(n)$ grid and with all edges parallel to the long diagonals of the grid cubes.
\item Any graph of maximum degree three can be drawn in 3d with optimal $120^\circ$ angular resolution with at most two bends per edge. However, our technique for achieving this small number of bends does not use a grid placement and does not achieve good volume bounds.
\item Any graph of maximum degree three has a drawing with $120^\circ$ angular resolution, integer vertex coordinates, edges parallel to the face diagonals of the integer grid, at most three bends per edge, and polynomial volume.
\end{itemize}
We believe that, as in the orthogonal case, it should be possible to achieve tighter bounds on the volume of the drawing at the expense of greater numbers of bends per edge.

\section{Three-bend drawings of degree-four graphs on a grid}
\label{sec:deg-four}

Our technique for three-dimensional drawings of degree-four graphs with angular resolution $109.5^\circ$ and three bends per edge is based on lifting two-dimensional drawings of the same graphs, with angular resolution $90^\circ$ and two bends per edge. The three-dimensional vertex placements are all on the plane $z=0$, essentially unchanged from their two-dimensional placements, but the edges are raised and lowered above and below the plane to avoid crossings and improve the angular resolution.

Our two-dimensional orthogonal drawing technique uses ideas from
previous work on drawing degree-four graphs with bounded geometric
thickness~\cite{DunEppKob-SCG-04}.  We begin by augmenting the graph
with dummy edges and a constant number of dummy vertices if necessary
to make it a simple 4-regular graph, find an Euler tour in the
augmented graph, and color the edges alternately red and green in
their order along this path. In this way, the red edges and the green
edges each form 2-regular subgraphs~\cite{Pet-AM-91} consisting of
disjoint unions of cycles. We denote the number of red (green) cycles
by $m_\red$ ($m_\green$). 

Next, we draw the red subgraph so that every cycle passes horizontally
through its vertices with two bends per edge, and we draw the green
subgraph so that every cycle passes vertically through its vertices
with two bends per edge. We can do that by using the cycle ordering
within each of these two subgraphs as one of the two Cartesian
coordinates for each point. More precisely, we do the
following. 

We define the \emph{green order} of the vertices of the graph to be an
order of the vertices such that the vertices of each green path or
cycle are consecutive; we define the \emph{red order} the same
way. Let $r_{\green}(v)\geq 0$ be the rank of a vertex $v$ in some
green order, and $r_{\red}(v)$ be its rank in some red order. We
further order the red and green cycles and define $c_{\red}(v) \geq 0$
and $c_{\green}(v) \geq 0$ to be the ranks in the two cycle orders of
the red and green cycles to which $v$ belongs. We embed the vertices
on a $(2n + 2m_\green - 4) \times (2n + 2m_\red - 4)$ grid such that
the $x$-coordinate of each vertex is $2r_{\green}(v) + 2c_\green(v)$,
and its $y$-coordinate is $2r_{\red}(v) + 2c_\red(v)$.

Let $v_1,... v_k$ be the vertices of a green cycle $C$ in the green
order. We embed $C$ as follows. We mark each end of each edge with a
plus or a minus such that at every vertex exactly one end is marked
with a plus and exactly one with a minus. We then would like to embed
$C$ in such a way that plus would correspond to the edge
entering the vertex from above and a minus corresponds to the edge
entering the vertex from below. Note that every edge whose two ends
are marked the same can be embedded in this way with two bends. Whenever
the marks alternate along the edge one can only embed it with two
bends if the lower end (the end incident to a vertex with smaller
$y$-coordinate) is marked with plus.

We next describe how to label $C$ so that it has a 2-bends-per-edge
embedding respecting the labeling. If $k$ is even, we mark both ends
of the edge $(v_1,v_2)$ with pluses. If $k$ is odd, we mark the higher
end of $(v_1,v_2)$ with a minus and its lower end with a plus. In both
cases there is a unique way to label the rest of the edges such that
both ends of each edge have the same signs and the labels alternate at
every vertex.

To complete our 2d embedding we draw all edges consistently with the
labeling as follows. Each edge $(v_i, v_{i+1})$ is placed such that
the $y$-distance of its horizontal segment to one of the vertices is
1. If the last edge $(v_1, v_k)$ is labeled negatively, its
horizontal segment is drawn on the grid line one unit below the
lowest vertex or bend of $C$. Similarly, if $(v_1, v_k)$ is labeled
positively, the horizontal segment is drawn one unit above the highest
part of $C$. See Figure~\ref{fig:2d-deg-4} for an illustration.

\begin{figure}[t]
\centering
\ifArxiv
\includegraphics[width=3.5in]{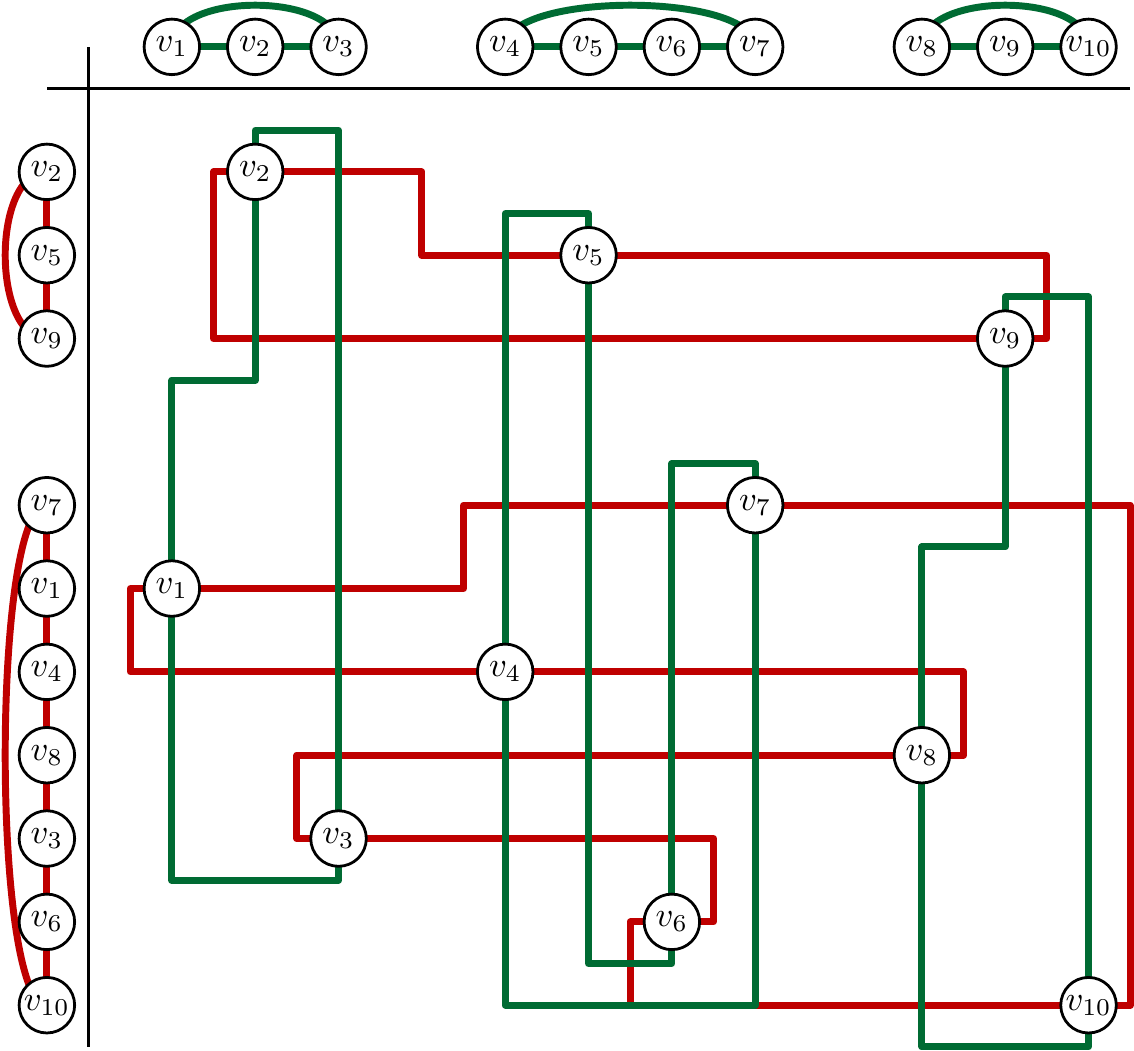}
\else
\includegraphics[width=2.25in]{example-deg4}
\fi
\caption{A 4-regular graph with 10 vertices embedded according to the
  decomposition into disjoint red and green cycles.}
\label{fig:2d-deg-4}
\end{figure}

\begin{lemma}\label{lem:2d-drawing}
  The embedding described above has the following
  properties:
  \begin{itemize}
  \item no two edges of the same color intersect;
  \item a vertex lies on an edge if and only if it is incident to the edge;
  \item no midpoint of an edge coincides with a bend of the edge;
  \item the embedding fits on a $(2n + 2m_\green) \times (2n + 2m_\red)$ grid.
  \end{itemize}
\end{lemma}

\begin{proof}
  Green edges connecting consecutive vertices in the green order of
  the same cycle $C$ are trivially disjoint. The horizontal segment of
  the edge connecting the first and the last vertex of $C$ is placed
  below or above all other edges of $C$. Two different green
  components are disjoint because the edges of every component are
  contained inside the vertical strip defined by its first and last
  vertices and components are ordered along the $x$-axis. The argument
  for red edges is symmetric.

  Since all the vertices have distinct $x$-coordinates, and every green
  vertical segment has a vertex at one of its ends we can conclude
  that every vertex is incident to at most two vertical green
  segments. Every green horizontal segment has odd $y$-coordinate
  and every vertex has even $y$-coordinate hence a green horizontal
  segment cannot contain a vertex. The argument for red edges is
  symmetric.


  For arbitrary red and green vertex orders it is possible that the
  midpoint of an edge coincides with one of its bends. We show that
  there are red and green vertex orders for which this is not the case. 
  For any edge whose ends are labeled differently we can always place
  the horizontal segment such that the midpoint of the edge does not
  coincide with a bend. For edges whose ends have the same label it is
  easy to see that the midpoint coincides with a bend if and only if
  the vertical distance and the horizontal distance of its vertices
  are equal. Apart from the last edge in each green cycle the
  horizontal distance between any two adjacent vertices $v_i$ and
  $v_{i+1}$ is 2. We claim that the vertical distance between $v_i$
  and $v_{i+1}$ is larger than 2 since otherwise $v_i$ and $v_{i+1}$
  are adjacent in a red cycle which contradicts the assumption that
  the 4-regular graph is simple. Note that this is the reason why
  different components are spaced by at least 4 units. Finally
  consider the last edge $(v_1,v_k)$ of a cycle $C$ with vertices
  $v_1, \dots, v_k$. The horizontal distance of $v_1$ and $v_k$ is
  $2k-2$. If their vertical distance equals $2k-2$ as well, we
  cyclically shift the green order of the vertices in $C$ by moving
  $v_k$ to the vertical grid line of $v_1$ and shifting each of $v_1,
  \dots, v_{k-1}$ two units to the right. Now $(v_{k}, v_{k-1})$ is
  the last edge of $C$. We perform this shifting until the vertices of
  the last edge no longer have vertical distance $2k-2$. Since every
  vertex has an exclusive $y$-coordinate there is at least one edge
  with this property in $C$. The local shifting of $C$ does not
  influence other parts of the drawing. The argument for red cycles is
  analogous.

  The vertices lie on $(2n+2m_\green-4) \times (2n+2m_\red-4)$ grid,
  and each grid line with coordinate $2k$ contains exactly one
  vertex. The lowest vertex is incident to a green edge with a
  horizontal segment at the height $-1$; the highest one is incident
  to a green edge with a horizontal segment at the height
  $2n+2m_\red-3$. One of the green edges connecting the first and last
  vertices of some cycle can lie one grid line below the height $-1$ or
  one grid lines above $2n+2m_\red-3$.
\end{proof}

It remains to lift the 2d drawing described above into three
dimensions. We first rotate the drawing by $45^\circ$; this expands
the grid size to $(4n+4m_\green)\times (4n+4m_\red)$. The vertices
themselves stay in the plane $z=0$, but we replace each edge by a path
in 3d that goes below the plane for the red edges and above the plane
for the green edges, eliminating all crossings between red and green
edges. The path for a green edge goes upwards along the long diagonals
of the diamond lattice cubes until its midpoint, where it has a bend
and turns downwards again. The lifted images of the two bends in the
underlying 2d edge remain bends in the 3d path and hence we get three
bends per edge in total. The red edges are drawn analogously below the
plane $z=0$. Since in the original 2d drawing every edge has even
length, the midpoint of every edge is a grid point and hence the
lifted midpoint is also a grid point of the diamond lattice. By
Lemma~\ref{lem:2d-drawing} a midpoint of an edge never coincides with
a 2d bend and hence all bend angles as well as the vertex angles are
$109.5^\circ$ diamond lattice angles.  Finally, we remove all the
edges we added to make the graph 4-regular. Considering the longest
possible red and green edges the total grid size is at most
$(4n+4m_\green)\times (4n+4m_\red)\times (12n+6m_\green+6m_\red)$.  We
note that $m_\green$, $m_\red \le n/3$ since every component is a
cycle. This yields the following theorem.

\begin{theorem}
  Any graph $G$ with maximum vertex degree four can be drawn in a 3d grid
  of size $16n/3 \times 16n/3 \times 16n$ with angular resolution
  $109.5^\circ$, three bends per edge and no edge crossings.
\end{theorem}


\section{Two-bend drawings of degree-three graphs} \label {sec:deg3-2bends}

The main idea of our algorithm for drawing degree-three graphs with optimal angular resolution and at most two bends per edge is to decompose the graph into a collection of vertex-disjoint cycles. Each cycle of length four or more can be drawn in such a way that the edges incident to the cycle all attach to it via segments that are parallel to the $z$ axis (Lemma~\ref{lem:single-cycle}). By placing the cycles far enough apart in the $z$ direction, these segments can be connected to each other with at most two bends per edge. However, several issues complicate this method:
\begin{itemize}
\item Cycles of length three cannot be drawn in the same way, and must be handled differently (Lemma~\ref{lem:triangle}).
\item Our method for eliminating cycles of length three does not apply to the graph $K_4$, for which we need a special-case drawing (Lemma~\ref{lem:k4}).
\item Although Petersen's theorem~\cite{BieBosDem-Algs-01,Pet-AM-91} can be used to decompose any bridgeless cubic graph into cycles and a matching, it is not suitable for our application because some of the matching edges may connect two vertices in a single cycle, a case that our method cannot handle. In addition, we wish to handle graphs that may contain bridges. Therefore, we need to devise a different decomposition algorithm. However, with our decomposition, the complement of the cycles is a forest rather than just a matching, and again we need additional analysis to handle this case.
\end{itemize}

\begin{figure}[t]
\centering\includegraphics[width=0.6\textwidth]{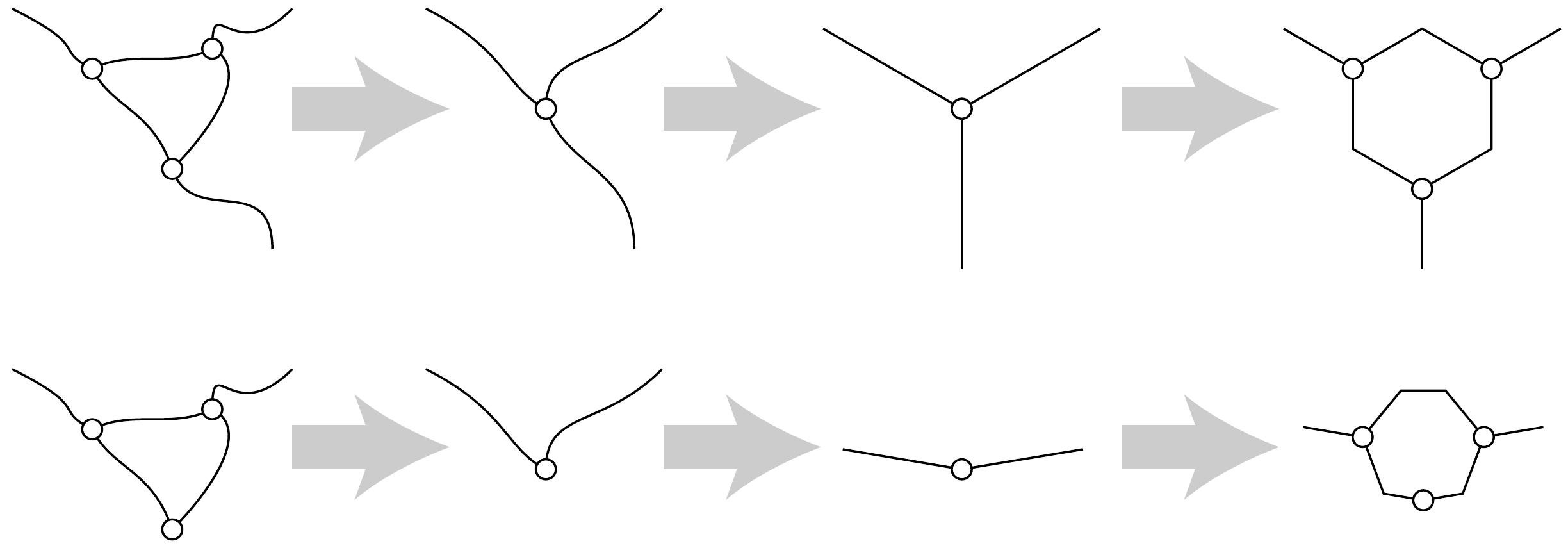}
\caption{$\Delta$--Y transformation of a graph $G$ containing a triangle, and undoing the transformation to find a drawing of $G$ (Lemma~\ref{lem:triangle}). Top: the contracted vertex has degree three, and is replaced by a hexagon. Bottom: the contracted vertex has degree two, and is replaced by a heptagon.}
\label{fig:DeltaY}
\end{figure}

\begin{lemma}
\label{lem:triangle}
Let $G$ be a graph with maximum degree three containing a triangle $uvw$. If $uvw$ is not part of any other triangle, let $G'$ be the result of contracting $uvw$ into a single vertex (that is, performing a $\Delta$--Y transformation on $G$). Otherwise, if there is a triangle $vwx$, let $G'$ be the result of contracting $uvwx$ into a single vertex.
If $G'$ can be drawn in 3d with two bends per edge and with angles of at least $120^\circ$ between the edges at each vertex or bend, then so can $G$.
\end{lemma}

\begin{proof}
First we consider the case that $G'$ is obtained by collapsing $uvw$.
The edges incident to the merged vertex $uvw$ must lie in a plane in
any drawing of $G'$. If $uvw$ has degree zero, one, or three in $G'$,
or if it has degree two and is drawn with angular resolution exactly
$120^\circ$, then we may draw $G$ by replacing $uvw$ by a small
regular hexagon in the same plane, with at most one bend for each of
the three triangle edges (Figure~\ref{fig:DeltaY}, top). If the merged
vertex $uvw$ has degree two in $G'$ and is drawn with angular
resolution greater than $120^\circ$, we may replace it by a small
heptagon (Figure~\ref{fig:DeltaY}, bottom).

The case that $G'$ is obtained by collapsing four vertices $uvwx$ is similar:
the collapsed vertex may be replaced by a pair of regular hexagons or irregular heptagons, meeting edge-to-edge. The four vertices $uvwx$ are placed at the points where these two polygons meet the other edges of the drawing and the two endpoints of the edge where they meet each other; the edge $vw$ has no bends and the other edges all have one or two bends.
\end{proof}

\begin{lemma}
\label{lem:k4}
The graph $K_4$ may be drawn in 3d with all vertices on integer grid points, angular resolution $120^\circ$, and at most two bends per edge.
\end{lemma}

\begin{proof}
See Figure~\ref{fig:k4}.
\begin{figure}[t]
\centering
\includegraphics[width=\textwidth]{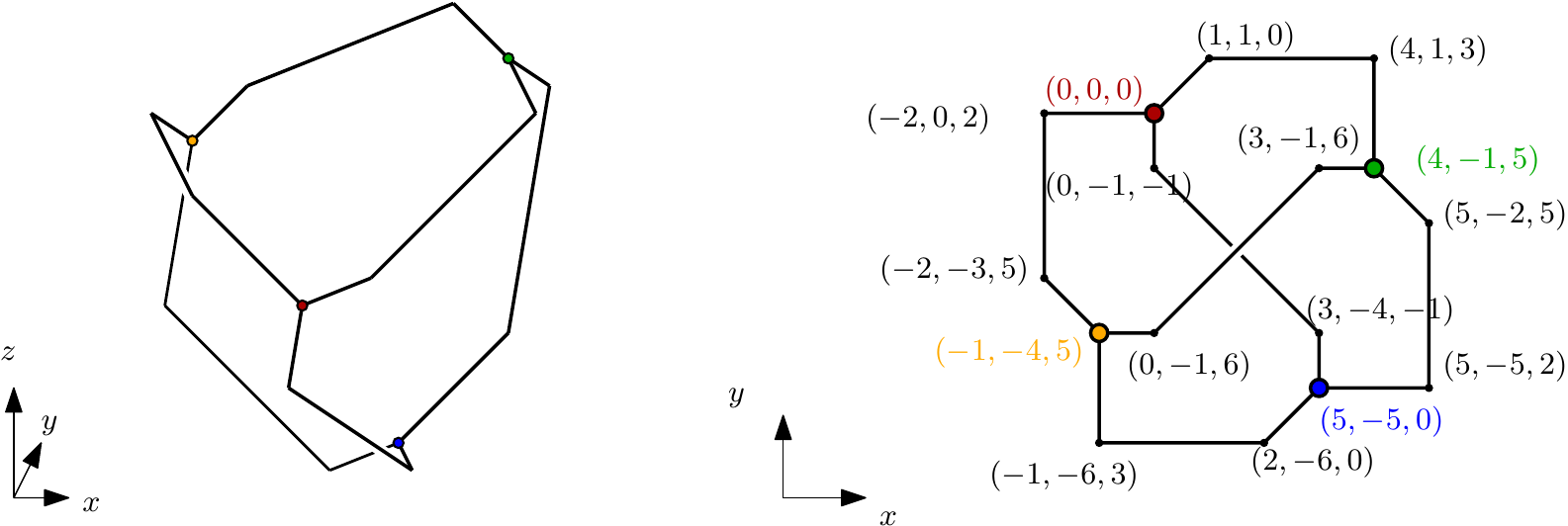}
\caption{A two-bend drawing of $K_4$ with $120^\circ$ angular resolution (left) and its two-dimensional projection (right).}
\label{fig:k4}
\end{figure}
\end{proof}

\begin{lemma}
\label{lem:single-cycle}
Let $G$ be a graph with maximum degree three, consisting of a cycle
$C$ of $n\ge 4$ vertices together with some number of degree-one
vertices that are adjacent to some of the vertices in $C$. Suppose
also that each degree-one vertex in $G$ is labeled with the number
$+1$ or $-1$. Then, there is a drawing of $G$ with the following
properties:
\begin{itemize}
\item All vertices and bends have angular resolution at least $120^\circ$.
\item All edges of $C$ have at most two bends.
\item All edges attaching the degree-one vertices to $C$ have no bends.
\item Every degree-one vertex has the same $x$ and $y$ coordinates as its (unique) neighbor, and its $z$ coordinate differs from its neighbor's $z$ coordinate by its label. Thus, all edges connecting degree-one vertices to $C$ are parallel to the $z$ axis, all positively labeled vertices are above (in the positive $z$ direction from) their neighbors, and all negatively labeled vertices are below their neighbors.
\item No three vertices of $C$ project to collinear points in the $(x,y)$-plane.
\end{itemize}
\end{lemma}

\begin{figure}[t]
\centering\includegraphics[height=1.25in]{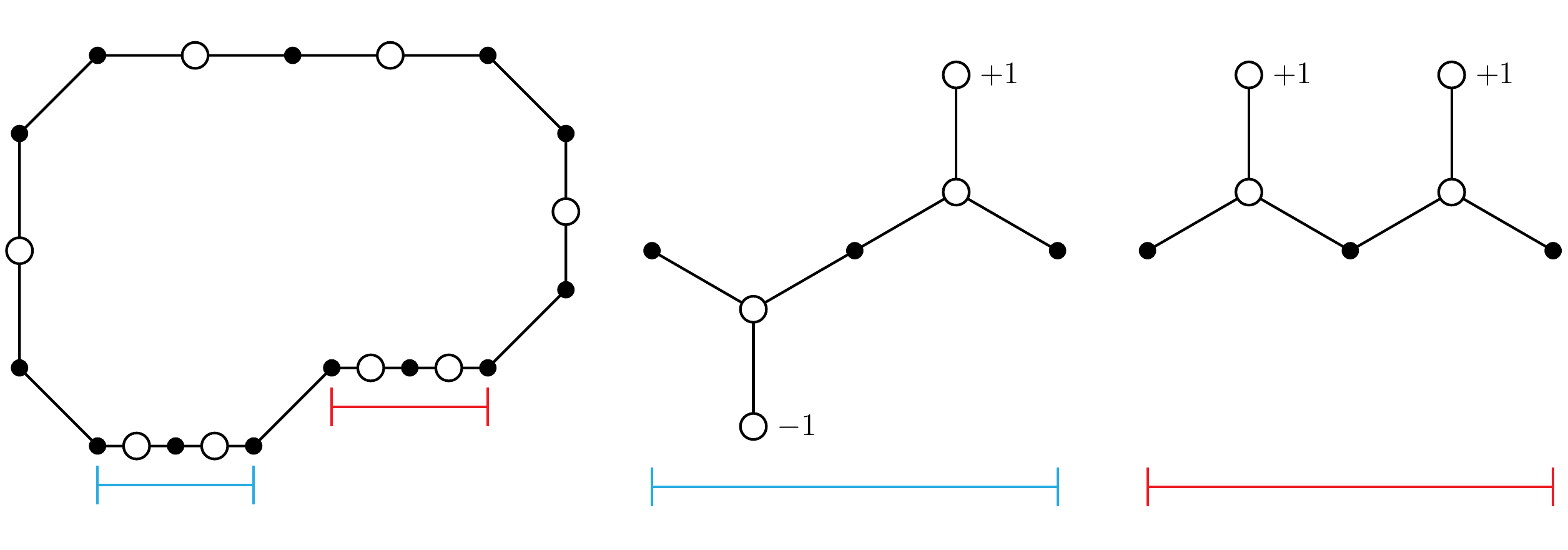}
\caption{The embedding of a cycle with degree-one neighboring vertices described by Lemma~\ref{lem:single-cycle}. Left: the $xy$-projection of the cycle; cycle vertices are indicated as large hollow circles and bends are indicated as small black disks. Right (at a larger scale): the $xz$-projection of the portions of the embedding corresponding to the two horizontal bottom sides of the $xy$-projected polygon.}
\label{fig:cycle120}
\end{figure}

\begin{proof}
As shown in Figure~\ref{fig:cycle120},
we draw $C$ in such a way that it projects onto a polygon $P$ in the $xy$-plane, with $135^\circ$ angles and with sides parallel to the coordinate axes and at $45^\circ$ angles to the axes. There are polygons of this type with a number of sides that can be any even number greater than seven; we choose the number of sides of $P$ so that at least one and at most two vertices of $C$ can be assigned to each axis-parallel side of the polygon.  (E.\,g., when $C$ has from four to eight vertices, $P$ can have eight sides, but when $C$ has more vertices $P$ must be more complex.)

We assign the vertices of $C$ consecutively to the axis-parallel sides
of $P$, in such a way that at least one vertex of $C$ and at most two
vertices are assigned to each axis-parallel side. If one vertex is assigned to a side, it is placed at the midpoint of that side, and if two vertices are assigned to a side of length $\ell$, then they are placed at distances of $\ell/4$ from one endpoint of the side, as measured in the $xy$ plane, with a bend at the midpoint of the side.

In three dimensions, the diagonal sides of $P$ are placed in the plane
$z=0$.
For any axis-parallel side of $P$ of length $\ell$ containing $k$ vertices of $C$,
we place the vertices with no degree-one neighbor or with a positively labeled
neighbor at elevation $z=\ell/(2k\sqrt{3})$, and the vertices with a
negatively labeled neighbor at elevation $z=-\ell/(2k\sqrt{3})$,
so that the portion of $C$ that projects onto a single side of $P$
forms a polygonal curve with angles of exactly $120^\circ$.
The degree-one
neighbors of the vertices in $C$ are then placed above or below them
according to their signs.

With this embedding, each vertex of $C$ gets angular resolution
exactly $120^\circ$. Any two consecutive vertices of $C$ that are
assigned to the same side of $P$ are separated either by zero bends
(if their neighbors have opposite signs) or a single bend (if their
neighbors have the same signs). Two consecutive vertices of $C$ that
belong to two different sides of $P$ are separated by two bends at two
of the corners of $P$; these bends have angles of
$\arccos(-\sqrt{3/8})\approx 127.8^\circ$. By adjusting the lengths of
the sides of $P$ appropriately, we may ensure that no three vertices
of $C$ project to collinear points in the $xy$-plane.
\end{proof}

The main idea of our drawing algorithm is to use Lemma~\ref{lem:single-cycle}, and some simpler cases for individual vertices, to repeatedly extend partial drawings of the given graph $G$ until the entire graph is drawn. We define a \emph{vertically extensible partial drawing} of a set $S$ of vertices of $G$ to be a drawing of the subgraph $G[S]$ induced in $G$ by $S$, with the following properties:
\begin{itemize}
\item The drawing of $G[S]$ has angular resolution $120^\circ$ or greater and has at most two bends per edge.
\item Each vertex in $S$ has at most one neighbor in $G\setminus S$.
\item If a vertex $v$ in $S$ has a neighbor $w$ in $G\setminus S$, then $w$ could be placed anywhere along a ray in the positive $z$-direction from $v$, producing a drawing of $G[S\cup\{w\}]$ that remains non-crossing, continues to have angular resolution $120^\circ$ or greater, and has no bends on edge $vw$. We call the ray from $v$ the \emph{extension ray} for edge $vw$.
\item No three extension rays are coplanar.
\end{itemize}

\begin{figure}[t]
\centering\includegraphics[height=1.25in]{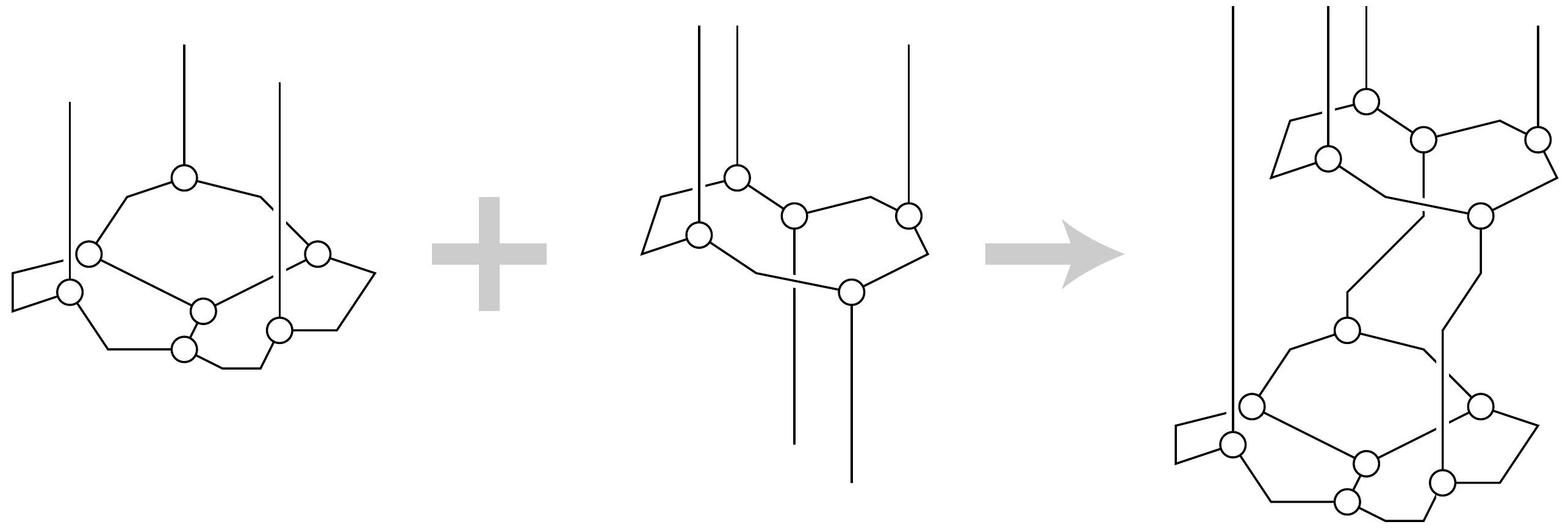}
\caption{Extending a vertically extensible drawing by adding a cycle.}
\label{fig:cyclex}
\end{figure}

For instance, if $C$ is a chordless cycle of length four or greater in $G$, then by Lemma~\ref{lem:single-cycle} there exists a vertically extensible partial drawing of $C$. More, the same lemma may be used to add another cycle to an existing vertically extensible partial drawing (Figure~\ref{fig:cyclex}):

\begin{lemma}
\label{lem:add-cycle}
For any vertically extensible drawing of a set $S$ of vertices in a graph $G$ of maximum degree three, and any chordless cycle $C$ of length four or more in $G\setminus S$, there exists a vertically extensible drawing of $S\cup C$.
\end{lemma}

\begin{proof}
For each vertex $v$ in $C$ that has a neighbor $w$ in $G$, replace $w$ with a degree-one vertex that has label $-1$ if $w\in S$ and $+1$ if $w\notin S$. Apply Lemma~\ref{lem:single-cycle} to find a drawing of $C$ that can be connected in the negative $z$-direction to the neighbors of $C$ in $S$, and in the positive $z$-direction for the remaining neighbors of $C$. Translate this drawing of $C$ in the $xy$-plane so that, among the extension rays of $S$ and the vertices of $C$, there are no three points and rays whose projections into the $xy$-plane are collinear and so that, when projected onto the $xy$-plane, the extension rays of $S$ (points in the $xy$-plane) are disjoint from the projection of the drawing of $C$.

For each extension ray of $S$ that connects a vertex $v$ of $S$ to a vertex $w$ in $C$, draw a two-bend path with $120^\circ$ bends in the plane containing the extension ray and $w$, such that the final segment of the path has the same $x$ and $y$ coordinates of $w$. By making the transverse section of this path be far enough away from $S$ in the positive $z$ direction, it will not intersect any other features of the existing drawing, and it cannot cross any of the other extension rays due to the requirement that no three of these rays be coplanar. If $C$ is translated in the positive $z$ direction farther than all of the bends in these paths, it can be connected to $S$ to form a vertically extensible drawing of $S\cup C$, as required.
\end{proof}

\begin{figure}[t]
\centering\includegraphics[scale=0.5]{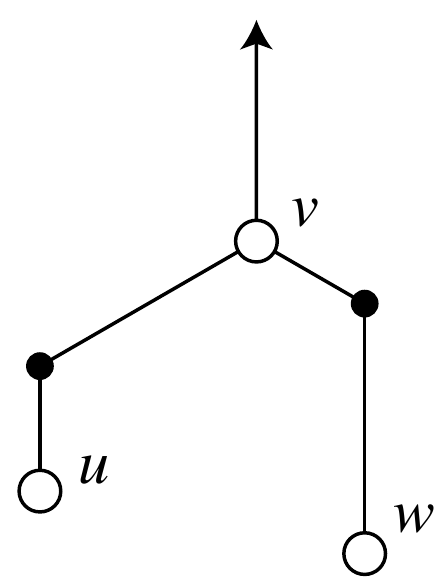}\qquad\includegraphics[scale=0.5]{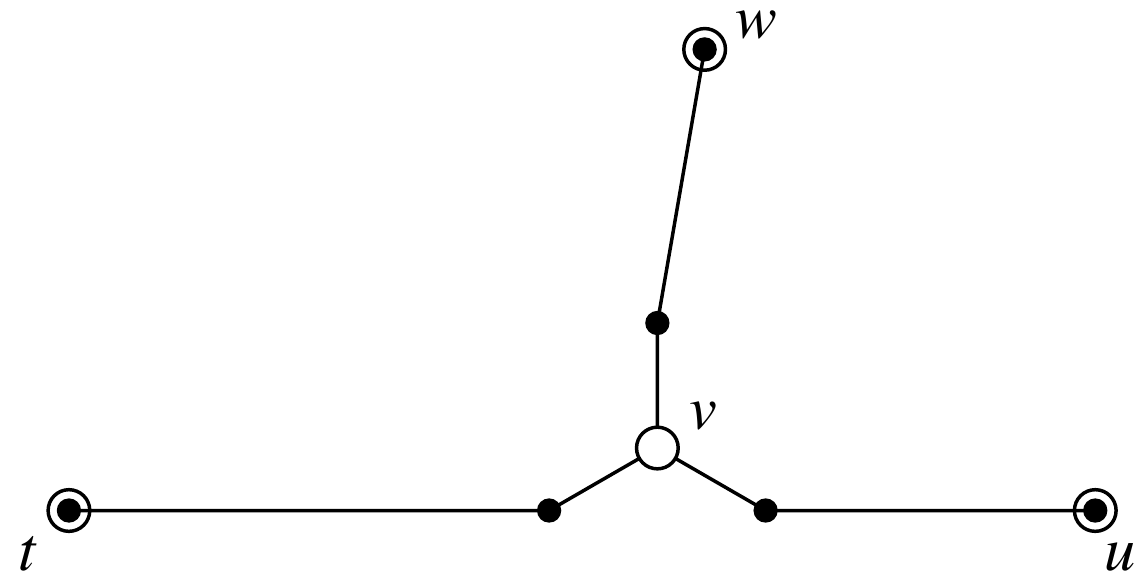}
\caption{Left: Adding a vertex $v$ with two neighbors $u$ and $w$ in $S$ and one neighbor in $G\setminus S$ to a vertically extensible drawing (shown in the plane of the extension rays of $u$ and $w$). Right: Adding a vertex $v$ with three neighbors $t$, $u$, and $w$ in $S$ (shown in the $xy$-plane). The three segments incident to $v$ are parallel to the $xy$ plane and the three remaining transverse segments form $120^\circ$ angles to the extension rays of $t$, $u$, and $w$. The bends where these transverse segments meet their extension rays are shown on top of the three points $t$, $u$, and $w$.}
\label{fig:extend1v}
\end{figure}

\begin{lemma}
\label{lem:add-vertex}
For any vertically extensible drawing of a set $S$ of vertices in a graph $G$ of maximum degree three, and any vertex $v$ in $G\setminus S$ with at most two neighbors in $S$ and at most one neighbor in $G\setminus S$, there exists a vertically extensible drawing of $S\cup \{v\}$.
\end{lemma}

\begin{proof}
If $v$ has no neighbors in $S$, then $v$ may be placed anywhere on any $z$-parallel line that does not pass through a feature of the existing drawing and is not coplanar with any two existing extension rays. If $v$ has a single neighbor $w$ in $S$, then $v$ may be placed anywhere on the extension ray of $wv$.

In the remaining case, $v$ connects to two extension rays of~$S$.
Within the plane of these two rays, we may connect $v$ to these two rays by transverse segments at $120^\circ$ angles to the rays. By placing $v$ far enough in the positive $z$ direction, these transverse segments can be made to avoid any existing features of the drawing. The extension ray from $v$ can lie on any line parallel to and between the lines of the two incoming extension rays; only finitely many of these lines lead to coplanarities with other extension rays, so it is always possible to place $v$ avoiding any such coplanarity. As shown in Figure~\ref{fig:extend1v}(left), this construction produces one bend on each edge into~$v$.
\end{proof}

\begin{lemma}
\label{lem:triskelion}
If we are given a vertically extensible drawing of a set $S$ of vertices in a graph $G$ of maximum degree three, and a vertex $v$ in $G\setminus S$ that has three neighbors $t$, $u$, and $w$ in $S$, then there exists a vertically extensible drawing of $S\cup \{v\}$.
\end{lemma}

\begin{proof}
  Suppose that $tu$ is the longest edge of the triangle formed by the
  projections of $t$, $u$, and $w$ into the $xy$ plane. Then, as a
  first approximation to the position of $v$ in the $xy$-plane, let
  the (two-dimensional) point $v'$ be placed on edge $tu$ of this
  triangle, at the point where $v'w$ is perpendicular to $tu$. We
  adjust this position along edge $tu$, keeping the angle between
  $v'w$ and $tu$ close to $90^\circ$ in order to ensure that line
  segment $v'w$ does not pass through
  the two-dimensional projection of any extension ray. Then, we
  replace $v'$ by three short line segments at $120^\circ$ angles to
  each other meeting the three line segments $v't$, $v'u$, and $v'w$
  at angles of $150^\circ$, $150^\circ$, and close to $180^\circ$. Let
  $v$ be the point where these three short line segments meet.

  This configuration can be lifted into three-dimensional space by
  placing $v$ and the three edges that attach to it in a plane
  perpendicular to the $z$ axis, and by replacing the remaining
  portions of line segments $v't$, $v'u$, and $v'w$ by transverse
  segments that make $120^\circ$ angles with the extension rays of
  $t$, $u$, and $w$. There are two bends per edge: one at the point
  where the extension ray of $t$, $u$, or $w$ meets a transverse
  segment, and one where a transverse segment meets one of the
  horizontal segments incident to $v$.

  The angles at the bends on the extension rays of $t$, $u$, and $w$
  are all exactly $120^\circ$, and the angles at the other bends on
  the paths connecting $t$ and $u$ to $w$ are $\arccos(3/4)\approx
  138.6^\circ$. As long as segment $v'w$ stays within $54^\circ$ of
  perpendicular to $tu$ in the $xy$-plane, the angle at the final
  remaining bend will be at least $120^\circ$.
\end{proof}

The construction of Lemma~\ref{lem:triskelion} is illustrated in Figure~\ref{fig:extend1v}(right).

\begin{theorem}
\label{thm:deg3bend2}
Any graph $G$ of degree three has a drawing with $120^\circ$ angular resolution and at most two bends per edge.
\end{theorem}

\begin{proof}
While $G$ contains a triangle, apply Lemma~\ref{lem:triangle} to simplify it, resulting in either $K_4$ or a triangle-free graph $G'$. If this simplification process leads to $K_4$,
draw it according to Lemma~\ref{lem:k4}. Otherwise, starting from
$S=\emptyset$, we repeatedly grow a vertically extensible drawing of a
subset $S$ of $G'$ until all of $G'$ has been drawn. If $G'\setminus
S$ contains a vertex with at most one neighbor in $G'\setminus S$, then either Lemma~\ref{lem:add-vertex} or Lemma~\ref{lem:triskelion} applies and we can add this vertex to the vertically extensible drawing. Otherwise, all vertices in $G'\setminus S$ have two or more neighbors in $G'\setminus S$, so $G'\setminus S$ contains a cycle. Let $C$ be the shortest cycle in $G'\setminus S$; it has length at least four (because we eliminated all triangles) and no chords (because a chord would lead to a shorter cycle) so we may apply Lemma~\ref{lem:add-cycle} to incorporate it into the vertically extensible drawing. Once we have included all vertices in the vertically extensible drawing, we have drawn all of $G'$, and we may reverse the transformations performed according to Lemma~\ref{lem:triangle} to produce a drawing of $G$.
\end{proof}

\ifArxiv

\section {Three-bend drawings of degree-three graphs on a grid}

In this section we provide an algorithm for embedding a degree-three
graph on a grid, using a similar approach to
Theorem~\ref{thm:deg3bend2} but with up to three instead of two bends
per edge. The grid will consist of the face diagonals of the cubes in
a regular grid of
cubes.
First of all, we will make a change of coordinates that allows us an
easier description. Define the $xy$-plane to be the plane spanned by
the edges $e_X = (0,1,1)$ and $e_Y = (1,1,0)$ and the $yz$-plane to be
the plane spanned by vectors $e_Y$ and $e_Z = (1,0,1)$. See
Figure~\ref{fig:skew}(a).

We will draw the different parts of the drawing in either a \emph
{horizontal} plane (parallel to the $xy$-plane) or in a \emph
{vertical} plane (parallel to the $yz$-plane).  The edges we use in
the $xy$-plane are parallel to an edge in the set $E_{XY} = \{e_X,
e_Y, e_X-e_Y\}$ and they all form angles of $120^\circ$.  Similarly,
in the $yz$-plane, all edges are parallel to an edge in $E_{YZ}=\{e_Y,
e_Z,e_Y-e_Z\}$. We will only use integer edge lengths.


\begin{figure}[h]
\centering\includegraphics[scale=0.8,page=2]{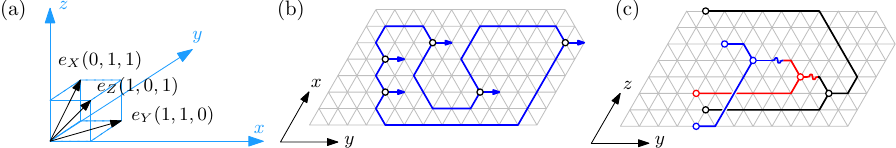}
\caption{(a) the three base vectors; (b) a cycle in a horizontal
  plane; (c) the tree structure connecting the extension rays in a set
  of three neighboring vertical planes (indicated by different colors).}
\label{fig:skew}
\end{figure}


The construction works similarly to the one described in Section~\ref
{sec:deg3-2bends}. In particular, we use exactly the same
decomposition of the graph into cycles and trees, and we still draw
every cycle in a different horizontal plane and extend the drawing in
$z$-direction with every new cycle. However, there are some important
differences. First of all, we no longer point the extension rays up
(in the $z$-direction), but to the right (the $y$-direction), within the plane in which we draw the cycle. As a result, the drawing
of a cycle is completely flat. Then, we draw the trees in vertical
planes through the extension rays of the respective
vertices. Figure~\ref {fig:skew}(b) and (c) shows the general idea.

\begin{lemma}
\label{lem:flat-single-cycle}
Let $G$ be a graph with maximum degree three, consisting of a cycle
$C$ of $n\ge 3$ vertices together with some number of degree-one
vertices that are adjacent to some of the vertices in $C$. Let $x_1,
\ldots, x_n$ be a set of distinct even integers bounded by $O(n)$. Then, there is a drawing
of $G$ in the $xy$-plane with the following properties:
\begin{itemize}
\item All vertices and bends have angular resolution $120^\circ$.
\item All edges of $C$ have at most three bends.
\item All edges of $G$ are parallel to an edge in $E_{XY}$.
\item For every degree-one vertex $v=(x_v,y_v,z_v)$ and its neighbor
  $u=(x_u,y_u,z_u)$ we have $x_v=x_u$, $y_v = y_u+1$, and $z_v=z_u$.
\item The drawing fits into a grid of size $O(n) \times O(n^2)$.
\item The $x$-coordinate of a vertex $v_i$ of $C$ is $x_i$, for all $1\leq i \leq n$.
\end{itemize}
\end{lemma}
\begin{proof}
We embed each cycle in a similar manner as in Section~\ref{sec:deg-four}.
We label each end of each edge of the cycle with a plus or a minus sign such
that at every vertex exactly one end is marked with a plus sign and exactly
one with a minus sign. We construct this labeling in exactly the same way as in
Section~\ref{sec:deg-four}.
Suppose that the $y$-axis of the $xy$-plane points
horizontally to the right. We then would like to embed the cycle in
the hexagonal grid of the $xy$-plane in such a way that edges labeled with a plus sign enter the vertex from above and edges labeled with a minus sign enter the vertex from below.
Moreover, the edge-segment entering a vertex from below is
$e_X$-parallel and the edge-segment entering a vertex from above is
$(e_X-e_Y)$-parallel. With this orientation for the edges, every edge whose ends have identical labels can
be embedded with exactly three bends. However, an edge that has opposite signs at its two ends can only be embedded with three bends if its
lower end (the end incident to a vertex with smaller $x$-coordinate)
is the end labeled with a plus sign.
We then embed the edges and vertices as follows.
\begin{itemize}
\item We place the vertex
$v_1$ at some point $(x_1,y_1)$ in the $xy$-plane, where $x_1$ is its
given $x$-coordinate.
\item If the cycle contains an
odd number of vertices then the first edge $(v_1,v_2)$ is labeled with two opposite signs, and is drawn as follows assuming that
$v_1$ is the lower of the two vertices. From $v_1$ we draw an
$e_X-e_Y$-parallel edge segment followed by an $e_X$-parallel segment
of equal length such that we reach the $e_Y$-parallel line at $x =
x_2$. We place $v_2$ at that position.
\item If $i>1$ or if $i=1$ and the cycle has an even number of vertices, then edge $v_i,v_{i+1}$ is labeled with two plus signs or two minus signs. In the case that it is labeled with two plus signs, we start drawing an $e_X-e_Y$-parallel edge segment from $v_i$
followed by an $e_X$-parallel segment of equal length until we reach
an $e_Y$-parallel line that is two units above the higher of the two
vertices. If these edge segments intersect any previous part
of the drawing we may need to spread the drawing in the $e_Y$-direction by
a distance of $O(n)$ that is added to the length of an $e_Y$-parallel
edge. From there we add a unit-length $e_Y$-parallel segment and
another $e_X-e_Y$-parallel segment until we reach the $e_Y$-parallel
line with $x=x_{i+1}$. This is where we place $v_{i+1}$.
We proceed
symmetrically for any edge marked with two minuses.
\item For the edge
$(v_1,v_n)$ the only difference is that its $e_Y$-parallel segment is
placed either two units below the lowest point or two units above the
highest point of the drawing according to its labels.
\end{itemize}
Finally, we embed each degree-one vertex one unit to the right of its
cycle-neighbor. See Figure~\ref{fig:skew}(b) for an illustration.
Note that the size of the grid for drawing $G$ is linear in the
$x$-direction but in the worst case quadratic in the
$y$-direction.
\end{proof}

As in our two-bend non-grid embedding for degree-three graphs, our overall embedding algorithm begins by finding and embedding a chordless
cycle of a given graph $G$ and then extends partial drawings of our graph
$G$ using Lemma~\ref{lem:flat-single-cycle} until we obtain the
drawing of the entire graph. We define an \emph{extensible partial
  grid-drawing} of a set $S$ of vertices of $G$ to be a crossing-free
grid-drawing of the subgraph $G[S]$ induced in $G$ by $S$, with the
following properties:
\begin{itemize}
\item The drawing of $G[S]$ has angular resolution $120^\circ$.
\item Each vertex in $S$ has at most one neighbor in $G\setminus S$.
\item Each vertex in $G\setminus S$ has at most one neighbor in $S$.
\item If a vertex $v$ in $S$ has a neighbor $w$ in $G\setminus S$,
  then we can draw an edge $(v,w)$ with at most three bends of
  $120^\circ$ that starts with an $e_Y$-parallel edge segment called
  the \emph{extension ray} of $v$. The placement of $(v,w)$ and $w$ is
  such that the resulting drawing of $G[S\cup\{w\}]$ has angular
  resolution $120^\circ$ and remains
  extensible and non-crossing.
\item For any $x$-coordinate $x_0$ there is at most one vertex $v$ in
  the vertical plane through $x_0$ with an \emph{active} extension
  ray, i.\,e., an extension ray that is not yet part of an actual edge
  $(v,w)$ since $w$ is still a vertex in $G\setminus S$.
\item All vertices in $S$ have even $z$-coordinates.
\end{itemize}

One difference between these properties and the ones used for our two-bend drawings is
the requirement that each vertex in $G\setminus S$ have at most one neighbor in $S$. To meet this requirement, when we add a cycle to the drawing, we will also add more vertices until this requirement is met.
To formalize this, define the \emph{double-adjacency closure} of a set
of vertices $S$ in a graph $G$ to be the smallest superset $W(S)
\supseteq S$ such that 
every vertex in $G\setminus W(S)$ that is adjacent to $W(S)$ has at least two other neighbors in $G \setminus W(S)$.
The double-adjacency closure of $S$ may be obtained by initializing a
variable set $W$ to be empty and then repeatedly adding to $W$
any vertex in $G\setminus (S \cup W)$ that has at most one neighbor in
$G\setminus (S \cup W)$ until no more such vertices exist; once this process
converges, the double-adjacency closure $W(S)$ is $S \cup W$.\martin{The way the
  closure is defined now is as a superset of $S$ hence I replaced all
  some occurrences of $S\cup W$ by $W(S)$.}

\begin{lemma}
\label{lem:add-flat-cycle}
For any extensible grid-drawing of a set $S$ of vertices in a graph
$G$ of maximum degree three, and any chordless cycle $C$ in
$G\setminus S$, there exists an extensible grid-drawing of $W(S\cup
C)$, where $W(S\cup C)$ is the double-adjacency closure of $S\cup C$.
\end{lemma}
\begin{proof}

  For each vertex $v$ in $C$ that has a neighbor $w$ in $G$, replace
  $w$ with a degree-one vertex. We next determine the $x$-coordinate
  $x_v$ for each vertex $v$ of $C$ as follows. If $v$ is adjacent to
  some vertex $w$ in $S$ or there is a third vertex $u$ in $G\setminus
  (S\cup C)$ that is adjacent to both $v$ and some vertex $w$ in $S$
  then we set $x_v$ to be the $x_w$ (we can do that because $w$ is
  unique). The same applies if there is a vertex $u$ in $G\setminus
  (S\cup C)$ that is adjacent to both $v$ and an already placed vertex
  $w$ of $C$. Otherwise we set $x_v$ to be the smallest even integer
  which is distinct from $x$-coordinates of all vertices in $S$ and
  all vertices in $C$ whose $x$-coordinates are already set. There is
  another special case to deal with: Let $u$ be a degree-three vertex
  in the double-adjacency closure $W(C)$ that has two neighbors in
  $W(C)$ and one neighbor $w$ in $S$. Then initially $u$ and all its
  predecessors in $W(C)$ are assigned a new $x$-coordinate. We need,
  however, that $u$ and all its predecessors are assigned the
  $x$-coordinate of $w$.

 \begin{figure}[tb]
    \centering
    \includegraphics[page=3]{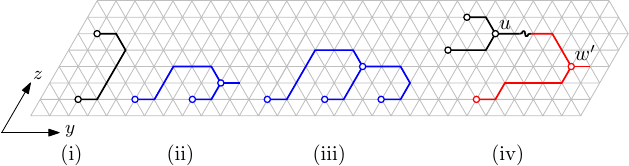}
    \caption{Different cases for connecting $S$ with a new cycle $C$
      (from left to right): i) an edge between two different cycles; ii)
    a vertex connected to two vertices of the same cycle; iii) a
    vertex connected to three vertices of the same cycle; iv) a vertex
  $u$ connected to two different cycles and then connected via $w$ to
  a third cycle involving a switch of vertical planes (indicated by $\sim$).}
    \label{fig:skew-cases}
  \end{figure}

  We will apply Lemma~\ref{lem:flat-single-cycle} to find a drawing of
  the double-adjacency closure $W(C)$ with the $x$-coordinates defined
  above such that $C$ is drawn in a horizontal plane. We place this horizontal plane high enough above the
  existing drawing of $S$ such that all the connections between $C$
  and $S$ can be drawn as 3-bend grid-paths with appropriate angular
  resolution. Note that since every cycle can spread as much as
  $O(n^2)$ in y-direction, we might have to draw the new cycle as far
  as $O(n^2)$ away from the existing drawing in z-direction. 
  If there is a vertex $u$ in $W(C)\setminus C$ that is adjacent
  to two vertices $v,w$ in $C$ then $v$ and $w$ lie in the same
  vertical plane. We place $u$ in that vertical plane and set its
  $z$-coordinate to be the next unused even $z$-coordinate above the
  drawing of $C$. Let $v$ have lower $y$-coordinate than $w$. Then we
  draw the edge $(v,u)$ in the vertical plane with three bends, where
  its four segments are $e_Y$-parallel, $e_Z$-parallel,
  $e_Y$-parallel, and $e_Y-e_Z$-parallel starting from $v$, see
  Figure~\ref{fig:skew-cases}(ii).  Thus
  $(v,u)$ connects to $u$ from above. We then draw the edge $(w,u)$
  with a single bend connecting to $u$ from below. If required we can
  place an extension ray for $u$ that is $e_Y$-parallel.  If there is
  a vertex $u$ in $W(C)\setminus C$ that is adjacent to three
  vertices $v,w,t$ in $C$ then by construction all three vertices have
  the same $x$-coordinate. Let $v,w,t$ be ordered by increasing
  $y$-coordinate. Then we place $u$ as if it had the two neighbors $v$
  and $w$ in $C$. Since $t$ is the rightmost vertex we can connect it
  with a three-bend edge to the extension ray of $u$ in the vertical
  plane of $u,v,w,t$, see Figure~\ref{fig:skew-cases}(iii). We note that
  this drawing of the double-adjacency closure $W(C)$ has indeed at
  most one active extension ray in each vertical plane.

  Now we connect $W(C)$ to $S$.  For each vertex $w$ of $S$ that is
  adjacent to a vertex $v$ in $W(C)$, draw a three-bend grid-path with
  $120^\circ$ bends in the plane containing the extension rays of $v$
  and $w$. Since $w$ and $v$ are placed at grid points with the same
  $x$-coordinates, this plane is parallel to the $yz$-plane and the
  edge $(v,w)$ follows the grid, see Figure~\ref{fig:skew-cases}(i).

  Similarly, if there is a vertex $u$ in $W(S\cup C)$ that
  is adjacent to a vertex $v$ in $C$ and a vertex $w$ in $S$ we have
  assigned identical $x$-coordinates to $v$ and $w$ and connect them
  in a vertical plane with three bends as if there was an edge
  $(v,w)$.  Then, however, we insert the vertex $u$ at the middle bend
  of the edge, and, if $u$ has degree three, add an $e_Y$-parallel
  extension ray to $u$.

  For any vertex $u$ in $W(S\cup C)$ that we introduce we need
  to check whether $u$ has a common neighbor $w'$ in $W(S\cup C)$
  together with another vertex
  $v'$ in $S\cup C$. If that is the case we also add $w'$ in the vertical
  plane of $v'$ as follows. The $x$-coordinates of $u$ and $v'$ do not
  match. However, since $u$ has an exclusive $z$-coordinate, we can
  spend two bends in the horizontal plane of $u$ to shift its
  extension ray to the $x$-coordinate of $v'$. Then we add $w'$ in the
  vertical plane of $v'$ so that the edge $(u,w')$ has three bends and
  the edge $(v',w')$ has at most three bends. This is illustrated in
  Figure~\ref{fig:skew-cases}(iv). We continue this process until all
  vertices in $W(S\cup C)$ are placed.

  Note that we never introduce crossings when drawing edges in
  vertical planes.  As we now show, there is at most one active extension
  ray in any vertical plane at any time. And since we extend the
  drawing in the positive $y$-direction this active extension ray is
  always rightmost in its vertical plane. By construction, this is the
  case for the existing drawings of $S$ and $W(C)$. Now we consider the
  combined drawing. It is certainly still true after we assign an
  unused $x$-coordinate $x_0$ to a vertex $v$. We will only assign the
  $x$-coordinate $x_0$ again if a new vertex $w$ connects either
  directly or via an intermediate vertex $u$ in $W(S\cup C)$
  to $v$. In the first case we draw the edge $(v,w)$ and have no more
  active extension rays. In the second case, we add the vertex $u$,
  which is then the only vertex with an active extension ray in this
  plane, and $u$ is to the right of $v$ and $w$. Whenever we use the
  active extension ray of a vertex $v$ to 
  connect to a new vertex $w$ then the $z$-coordinate of $w$ is larger
  than the $z$-coordinate of any existing point in that vertical
  plane. So we connect the rightmost point with the topmost point
  in the vertical plane, and hence the new edge does not produce
  crossings.
\end{proof}

\begin{figure}[h]
\centering\includegraphics{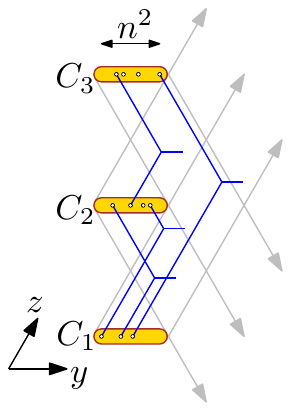}
\caption{Connections between cycles $C_i$ and $C_j$ can be forced to immediately move their horizontal planes (because of crossings within that plane), so they have to be spaced enough to make any connection possible. Since one cycle has size $n^2$, their spacing is $O(n^2)$ and the total distance in the $z$-direction is $O(n^3)$, and the resulting distance in the $y$-direction is therefore also $O(n^3)$.}
\label{fig:global}
\label{fig:skew-size}
\end{figure}

The next theorem follows from Lemma~\ref{lem:add-flat-cycle} and Figure~\ref{fig:skew-size}.

\begin{theorem}
  Any graph $G$ with maximum vertex degree three can be drawn in a 3d grid
  of size $O(n^3)\times O(n^3)\times O(n^3)$ with angular resolution
  $120^\circ$, three bends per edge and no edge crossings.
\end{theorem}
\begin{proof}
  As in our two-bend drawing of the same graphs, we decompose the
  graph into a sequence of cycles and isolated vertices, where each
  isolated vertex belongs to the doubly-adjacent closure of the
  previous cycles. We apply Lemma~\ref{lem:add-flat-cycle} to extend
  the drawing for each successive cycle.  We start by drawing the
  first cycle in the $z = 0$ plane, and extend the drawing by always
  adding new cycles above the first one and drawing trees extending
  into the $y$-direction. Our drawing uses $O(n)$ different $yz$
  planes, and therefore has extent $O(n)$ in the $x$ direction of our
  modified coordinate system. As the analysis in
  Figure~\ref{fig:skew-size} shows, it extends for $O(n^3)$ units in
  the other two coordinate directions. This leads to a final drawing
  of size $O(n) \times O(n^3) \times O(n^3)$ in our modified
  coordinate system.

  Finally, if we change the coordinate system back to the more
  standard Cartesian coordinates, this means that in each dimension
  the size can be $O(n^3)$.
\end{proof}


\else
In the full version we show that any graph of degree three has a drawing with $120^\circ$ angular resolution, integer vertex coordinates, edges parallel to the face diagonals of the integer grid, at most three bends per edge, and polynomial volume.
\fi
\section{Conclusions}

We have shown how to draw degree-three graphs in three dimensions with
optimal angular resolution and two bends per edge, and how to draw
degree-four graphs in three dimensions with optimal angular
resolution, three bends per edge, integer vertex coordinates, and
cubic volume.  Multiple questions remain open for investigation,
however:
\begin{itemize}
\item It does not seem to be possible to draw $K_4$ or $K_5$ in three dimensions with optimal angular resolution and one bend per edge. Can this be proven rigorously?
\item Does every degree-four graph have a drawing in three dimensions with optimal angular resolution and two bends per edge? In particular, is this possible for $K_5$?
\item How many bends per edge are necessary to draw degree-three graphs with optimal angular resolution in an $O(n)\times O(n)\times O(n)$ grid, with all edges parallel to the face diagonals of the grid?
\item It should be possible to draw degree-three and degree-four graphs with optimal angular resolution in an $O(\sqrt n)\times O(\sqrt n)\times O(\sqrt n)$ grid. How many bends per edge are necessary for such a drawing?
\end{itemize}
\ifProceedings
\smallskip
\noindent\textbf{Acknowledgments.}
\else
\subsection*{Acknowledgments}
\fi
This research was supported in part by the National Science
Foundation under grant 0830403, by the
Office of Naval Research under MURI grant N00014-08-1-1015, and by the
German Research Foundation (DFG) under grant NO 899/1-1.

{ \ifArxiv \raggedright \fi
  \bibliographystyle{abuser}
  \bibliography{ar3d}

\begin{thebibliography}{10}

\bibitem{AngCitDiB-GD-09}
P.~Angelini, L.~Cittadini, G.~Di~Battista, W.~Didimo, F.~Frati, M.~Kaufmann,
  and A.~Symvonis.
\newblock {On the perspectives opened by right angle crossing drawings}.
\newblock {\em Proc. 17th Int. Symp. Graph Drawing (GD 2009)}, pp.~21{--}32.
  Springer-Verlag, LNCS 5849, 2010,
  \href{http://dx.doi.org/10.1007/978-3-642-11805-0\_5}%
{doi:10.1007/978-3-642-11805-0\_5}.

\bibitem{BieBosDem-Algs-01}
T.~C. Biedl, P.~Bose, E.~D. Demaine, and A.~Lubiw.
\newblock {Efficient algorithms for Petersen's matching theorem}.
\newblock {\em Journal of Algorithms} 38(1):110{--}134, 2001,
  \href{http://dx.doi.org/10.1006/jagm.2000.1132}%
{doi:10.1006/jagm.2000.1132}.

\bibitem{BieThiWoo-Algo-06}
T.~C. Biedl, T.~Thiele, and D.~R. Wood.
\newblock {Three-dimensional orthogonal graph drawing with optimal volume}.
\newblock {\em Algorithmica} 44(3):233{--}255, 2006,
  \href{http://dx.doi.org/10.1007/s00453-005-1148-z}%
{doi:10.1007/s00453-005-1148-z}.

\bibitem{CarEpp-GD-06}
J.~Carlson and D.~Eppstein.
\newblock {Trees with convex faces and optimal angles}.
\newblock {\em Proc. 14th Int. Symp. Graph Drawing (GD 2006)}, pp.~77{--}88.
  Springer-Verlag, LNCS 4372, 2007,
  \href{http://dx.doi.org/10.1007/978-3-540-70904-6\_9}%
{doi:10.1007/978-3-540-70904-6\_9},
  \href{http://arxiv.org/abs/cs.CG/0607113}{arXiv:cs.CG/0607113}.

\bibitem{ClaKep-PRSL-86}
B.~W. Clare and D.~L. Kepert.
\newblock {The closest packing of equal circles on a sphere}.
\newblock {\em Proc. Roy. Soc. London A} 405(1829):329{--}344, 1986,
  \href{http://dx.doi.org/10.1098/rspa.1986.0056}%
{doi:10.1098/rspa.1986.0056}.

\bibitem{Dem-TOPP-02}
E.~D. Demaine.
\newblock {Problem 46: 3D minimum-bend orthogonal graph drawings}.
\newblock The Open Problems Project, 2002,
  \url{http://maven.smith.edu/~orourke/TOPP/P46.html}.
\newblock Posed by David R. Wood at the CCCG 2002 open-problem session.

\bibitem{DidEadLio-WADS-09}
W.~Didimo, P.~Eades, and G.~Liotta.
\newblock {Drawing graphs with right angle crossings}.
\newblock {\em Proc. 11th Int. Symp. Algorithms and Data Structures (WADS
  2009)}, pp.~206{--}217. Springer-Verlag, LNCS 5664, 2009,
  \href{http://dx.doi.org/10.1007/978-3-642-03367-4\_19}%
{doi:10.1007/978-3-642-03367-4\_19}.

\bibitem{DujGudMor-09}
V.~Dujmovi{\'c}, J.~Gudmundsson, P.~Morin, and T.~Wolle.
\newblock {Notes on large angle crossing graphs},
  \href{http://arxiv.org/abs/0908.3545}{arXiv:0908.3545}.
\newblock 2009.

\bibitem{DunEppKob-SCG-04}
C.~A. Duncan, D.~Eppstein, and S.~G. Kobourov.
\newblock {The geometric thickness of low degree graphs}.
\newblock {\em Proc. 20th ACM Symp. Computational Geometry (SoCG 2004)},
  pp.~340{--}346, 2004, \href{http://dx.doi.org/10.1145/997817.997868}%
{doi:10.1145/997817.997868},
  \href{http://arxiv.org/abs/cs.CG/0312056}{arXiv:cs.CG/0312056}.

\bibitem{EadSymWhi-DAM-00}
P.~Eades, A.~Symvonis, and S.~Whitesides.
\newblock {Three-dimensional orthogonal graph drawing algorithms}.
\newblock {\em Discrete Applied Mathematics} 103(1{--}3):55{--}87, 2000,
  \href{http://dx.doi.org/10.1016/S0166-218X(00)00172-4}%
{doi:10.1016/S0166-218X(00)00172-4}.

\bibitem{EigFekKla-DG-01}
M.~Eiglsperger, S.~P. Fekete, and G.~W. Klau.
\newblock {Orthogonal graph drawing}.
\newblock {\em Drawing Graphs: Methods and Models}, pp.~121{--}171.
  Springer-Verlag, LNCS 2025, 2001,
  \href{http://dx.doi.org/10.1007/3-540-44969-8\_6}%
{doi:10.1007/3-540-44969-8\_6}.

\bibitem{Epp-GD-08}
D.~Eppstein.
\newblock {Isometric diamond subgraphs}.
\newblock {\em Proc. 16th Int. Symp. Graph Drawing (GD 2008)}, pp.~384{--}389.
  Springer-Verlag, LNCS 5417, 2009,
  \href{http://dx.doi.org/10.1007/978-3-642-00219-9\_37}%
{doi:10.1007/978-3-642-00219-9\_37}.

\bibitem{GarTam-ESA-94}
A.~Garg and R.~Tamassia.
\newblock {Planar drawings and angular resolution: algorithms and bounds}.
\newblock {\em Proc. 2nd Eur. Symp. Algorithms (ESA 1994) (LNCS)} 855:12{--}23,
  1994, \href{http://dx.doi.org/10.1007/BFb0049393}%
{doi:10.1007/BFb0049393}.

\bibitem{GutMut-GD-97}
C.~Gutwenger and P.~Mutzel.
\newblock {Planar polyline drawings with good angular resolution}.
\newblock {\em Proc. 6th Int. Symp. Graph Drawing (GD 1998)}, pp.~167{--}182.
  Springer-Verlag, LNCS 1547, 1998,
  \href{http://dx.doi.org/10.1007/3-540-37623-2\_13}%
{doi:10.1007/3-540-37623-2\_13}.

\bibitem{HuaHonEad-PacVis-08}
W.~Huang, S.-H. Hong, and P.~Eades.
\newblock {Effects of crossing angles}.
\newblock {\em Proc. IEEE Pacific Visualization Symp.}, pp.~41{--}46, 2008,
  \href{http://dx.doi.org/10.1109/PACIFICVIS.2008.4475457}%
{doi:10.1109/PACIFICVIS.2008.4475457}.

\bibitem{Mal-STOC-92}
S.~Malitz.
\newblock {On the angular resolution of planar graphs}.
\newblock {\em Proc. 24th ACM Symp. Theory of Computing (STOC 1992)},
  pp.~527{--}538, 1992, \href{http://dx.doi.org/10.1145/129712.129764}%
{doi:10.1145/129712.129764}.

\bibitem{Pet-AM-91}
J.~Petersen.
\newblock {Die Theorie der regul{\"a}ren Graphs}.
\newblock {\em Acta Math.} 15(1):193{--}220, 1891,
  \href{http://dx.doi.org/10.1007/BF02392606}%
{doi:10.1007/BF02392606}.

\bibitem{Tam-RTBN-30}
P.~M.~L. Tammes.
\newblock {On the origin of the number and arrangement of the places of exit on
  the surface of pollen grains}.
\newblock {\em Ree. Trav. Bot. N{\'e}erl.} 27:1{--}82, 1930.

\bibitem{Woo-TCS-03}
D.~R. Wood.
\newblock {Optimal three-dimensional orthogonal graph drawing in the general
  position model}.
\newblock {\em Theoretical Computer Science} 299(1-3):151{--}178, 2003,
  \href{http://dx.doi.org/10.1016/S0304-3975(02)00044-0}%
{doi:10.1016/S0304-3975(02)00044-0}.

\end{thebibliography}
}

\ifSubmission
\appendix

\fi

\end{document}